\documentclass[11pt, epsfig]{article}
\usepackage{epsfig, amsmath, amssymb, amsthm, times}
\usepackage{graphicx}
\usepackage{color}
\usepackage{bm}

\def\be{\begin{equation}}
\def\ee{\end{equation}}

\def\1{\mathbf{1}}

\newtheorem{theorem}{Theorem}[section]

\newtheorem{lemma}[theorem]{Lemma}

\newtheorem{remark}[theorem]{Remark}

\title{Stability of twisted states on lattices of Kuramoto oscillators}

\author{Monica Goebel\thanks{Department of 
		Mathematics and Statistics,
		The College of New Jersey} \and Matthew S. Mizuhara$^{*,}$\thanks{Corresponding author:
		{\tt  mizuharm@tcnj.edu}} \and Sofia Stepanoff$^*$}
	
\begin{document}

\maketitle

	\begin{abstract}
		Real world systems comprised of coupled oscillators have the ability to exhibit spontaneous synchronization and other complex behaviors. The interplay between the underlying network topology and the emergent dynamics remains a rich area of investigation for both theory and experiment. In this work we study lattices of coupled Kuramoto oscillators with non-local interactions. Our focus is on the stability of twisted states. These are equilibrium solutions with constant phase shifts between oscillators resulting in spatially linear profiles. Linear stability analysis follows from studying the quadratic form associated with the Jacobian matrix. Novel estimates on both stable and unstable regimes of twisted states are obtained in several cases. Moreover, exploiting the ``almost circulant'' nature of the Jacobian obtains a surprisingly accurate numerical test for stability. While our focus is on 2D square lattices, we show how our results can be extended to higher dimensions.
	\end{abstract}

	\section{Introduction}
	
	Coupled systems of oscillators are ubiquitous in physical, chemical, and biological settings, including collections of synchronous fireflies, cardiac pacemakers, neuron networks, and electrical power grids. Such systems often exhibit spontaneous emergence of large scale synchronization and other complex collective effects despite inherent heterogeneity among the individual components \cite{strogatz2012sync, pikovsky2003synchronization}. The Kuramoto model is perhaps the most widely studied mathematical model describing such systems \cite{kuramoto1975self}. It is comprised of phase oscillators which are coupled attractively through potentially complex networks of interactions. It has led to extremely fruitful research in various fields since it is able to capture the onset of synchronization and other complex patterns such as chimeras and clusters, yet it is simple enough to be amenable to rigorous analysis revealing underlying mechanisms generating such states \cite{rodrigues2016kuramoto, acebonluivicperfelspi05, strogatz2000kuramoto}.
	
	In this work we study coupled Kuramoto oscillators on square lattices. From a modeling point of view, lattices are very natural to study as they account for 2D and 3D spatial organization among oscillators. Since the pioneering work of Sakaguchi, Shinomoto, and Kuramoto \cite{sakaguchi1987local}, there have been many studies of coupled oscillator systems on lattices investigating synchronization \cite{strogatz1988phase, daido1988lower, hong2005collective, laing2011fronts, aoyagi1991frequency, ostborn2009renormalization, migliorini2008critical, belykh2003cluster, belykh2003persistent, corral1995self}, spiral patterns \cite{ottino2016frequency, lee2010vortices, paullet1994stable, sarkar2021phase, bukh2019spiral, totz2018spiral}, and chimeras \cite{shima2004rotating, laing2009dynamics,shepelev2020quantifying, hagerstrom2012experimental,shepelev2018double, nkomo2013chimera,omel2012stationary}.
	
	Our focus is on the analysis of twisted state solutions of Kuramoto oscillators on lattices. Twisted states (sometimes called splay states) are equilibrium solutions which exhibit constant phase shifts between neighboring oscillators, resulting in spatially linear patterns. Such states can represent, e.g., traveling waves of activity in spatially extended systems. Twisted states have been investigated for coupled oscillator systems on various networks \cite{omel2014partially, chiba2018bifurcations, medtan15, medvedev2014small, girnyk2012multistability, deville2019synchronization, lee2018twisted, wiley2006size, bolotov2019twisted,medvedev2021chimeras}, however the study of twisted states on lattices remains relatively unexplored \cite{salova2020decoupled, lee2018twisted}. Our work primarily extends the results of Lee, Cho, and Hong \cite{lee2018twisted} where twisted states on nearest neighbor coupled 2D and 3D lattices were investigated. We study more general types of twisted state on lattices with longer range interactions.	
	
	After introducing the model, we prove sharp estimates in the case of nearest neighbor coupling. Next, we show how to extend our techniques to include lattices with additional coupling from diagonal interactions. We obtain sufficient conditions for both stability and instability, however unlike the nearest neighbor case, these estimates are no longer sharp. We then show how our work naturally extends to arbitrary interaction ranges and in general dimensions. Finally, to show how one can hope to sharpen estimates in the non-local cases, we present an approximation of the eigenvalues of the Jacobian, which does an impressively accurate job of predicting both regions of stability and instability. This test is based on recognizing the Jacobian matrix as ``almost circulant.''

	%
	%
	
	\subsection{Model} 
	
	We first recall the Kuramoto model of coupled oscillators,
	\begin{equation}\label{eq:kuramoto}
	\dot{u}_{i} = \omega_{i} + K\sum_{k=1}^n A_{i,k} \sin(u_k-u_i),\;\; i\in \{1,\dots, n\},
	\end{equation}
	where $u_i \colon [0,T]\to \mathbb{S}^1$ is the phase of the $i$th oscillator, $\omega_i$ is its natural frequency, $K$ the coupling strength, and $A=(A_{i,j})$ is the adjacency matrix encoding the network topology. We begin by assuming a 2D lattice of oscillators so that $n = L\times L$. For simplicity we work with identical oscillators: $\omega_i\equiv \omega$. Then we can move into a rotating frame $u_i+\omega t$ and rescale time so that $K=1$.	
	
	To build the adjacency matrix, we introduce the convenient relabeling 
	\begin{equation*}
	\theta_{i,j} = u_{(i-1)L + j},\;\; i,j\in\{1,\dots, L\}
	\end{equation*}
	so as to translate between a vector of oscillators $u_i$ and an array of oscillators $\theta_{i,j}$ so that \eqref{eq:kuramoto} can be written as
	\begin{equation}\label{eq:kuramoto1}
	\dot{\theta}_{i,j} = \sum_{j=1}^L \sum_{k=1}^L A_{(i-1)L+j,(k-1)L+\ell} \sin(\theta_{k,\ell}-\theta_{i,j}).
	\end{equation}
	
	Writing $\theta_{i,j}\sim \theta_{k,\ell}$ to mean that two oscillators are adjacent, then {\em $r$-nearest neighbor} coupling is given by	
	\begin{equation}
	\theta_{i,j} \sim \theta_{k,\ell} \mbox{ if and only if } 0<\|(i,j)-(k,\ell)\|_{T^2} \leq r^2	
	\end{equation}
	where
	\begin{equation*}
	\|(i,j)-(k,\ell)\|_{T^2} := \min\{|i-k|,L-|i-k|\}^2 + \min\{|j-\ell|,L-|j-\ell|\}^2
	\end{equation*}
	is the Euclidean distance on the 2D array with periodic boundaries (e.g., a flat torus). It is evident that all connections are symmetric, resulting in an adjacency matrix $A$ which is symmetric. Figure \ref{fig:lattice} shows the coupling for several representative values of $r$.
	
	\begin{figure}
		\centering
		{\bf a}\includegraphics[width = .25\textwidth]{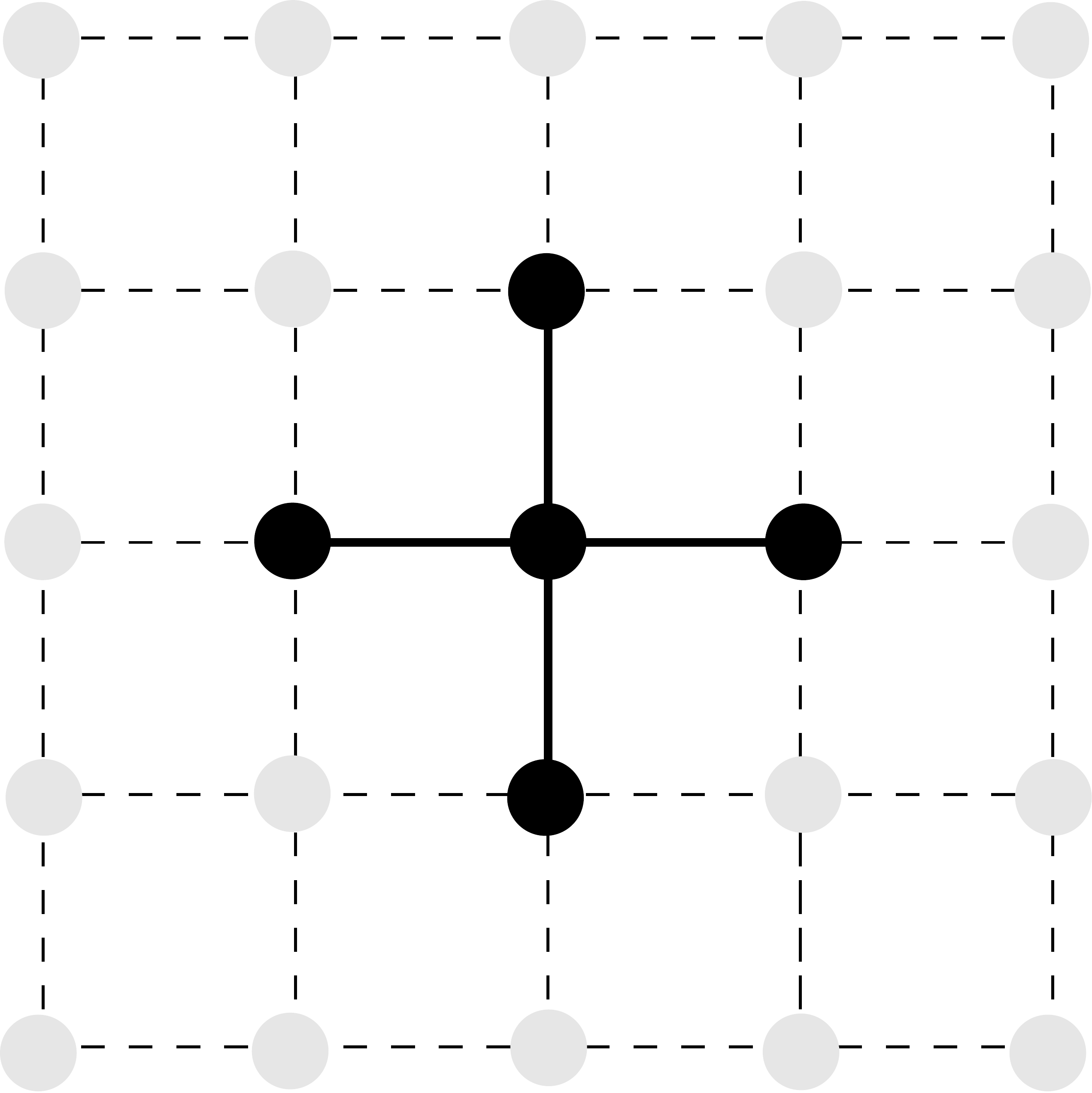}
		\hspace*{3mm}
		{\bf b}\includegraphics[width = .25\textwidth]{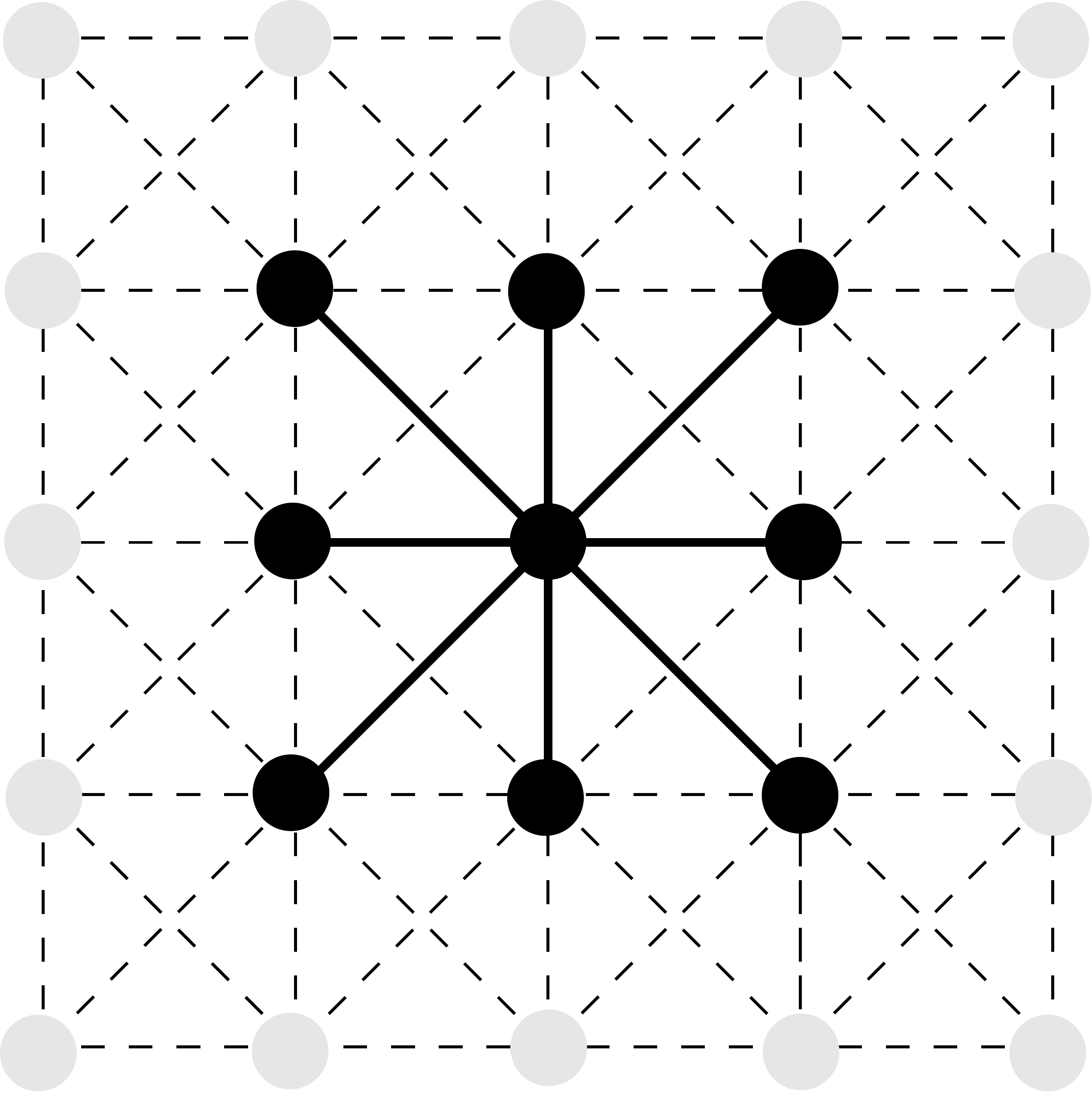}
		\hspace*{3mm}
		{\bf c}\includegraphics[width = .25\textwidth]{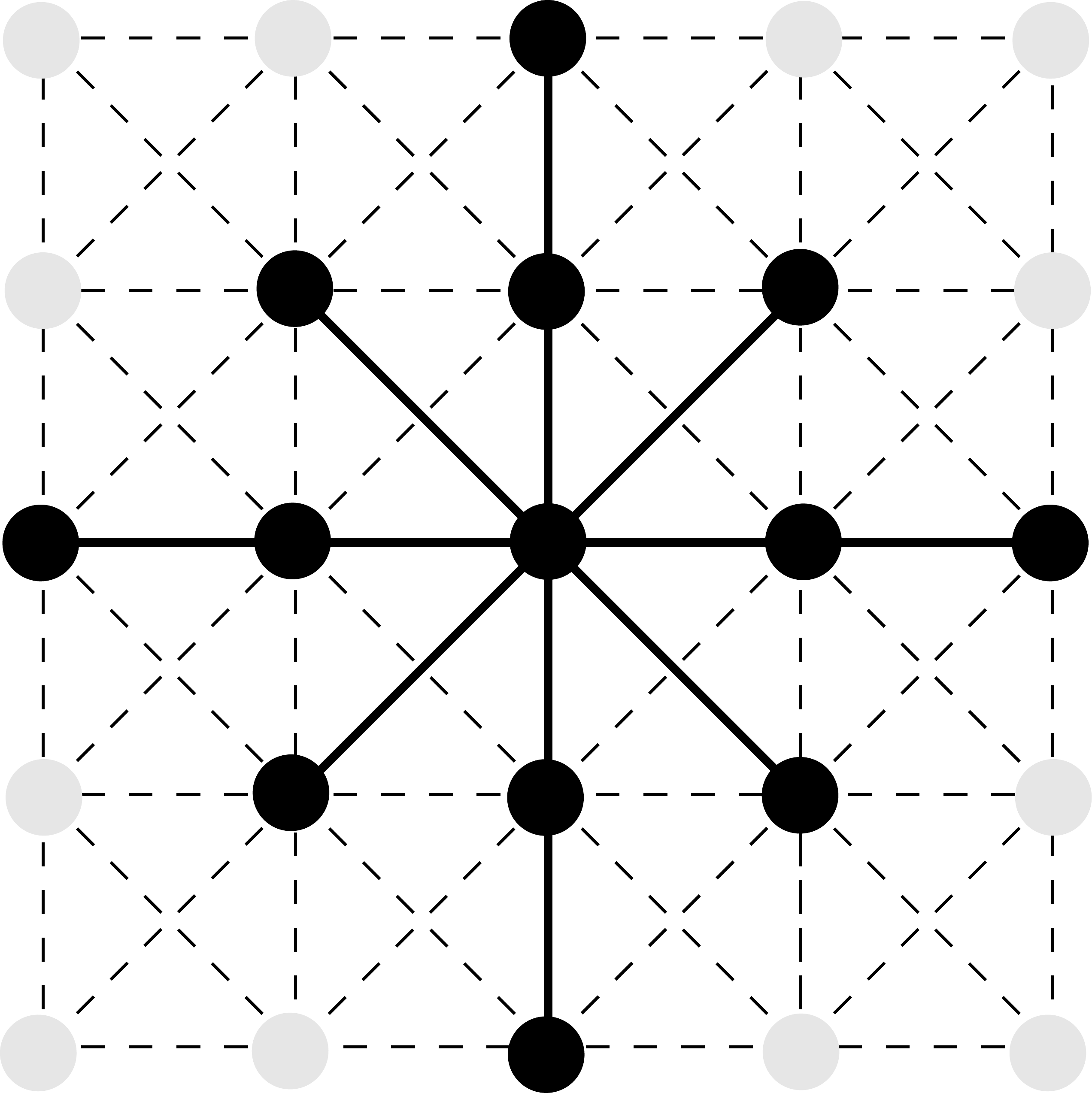}
		\caption{Diagram of $r$-nearest neighbor coupling on a 2D array for ({\bf a}) $r=1$, ({\bf b}) $r=\sqrt{2}$, ({\bf c}) $r=2$. Each diagram highlights the connections for the center oscillator only, and we omit drawing any connections on the boundary which must ``wrap'' around the figure.}
		\label{fig:lattice}
	\end{figure}

	\section{Stability of twisted states}
	
	Let $q_1,q_2\in \mathbb{Z}$. A $(q_1,q_2)$-twisted state is an equilibrium solution of the form
	\begin{equation}
	\theta_{i,j} = \frac{2\pi q_1 i }{L} + \frac{2\pi q_2 j}{L} + C \pmod{2\pi},
	\end{equation}
	where we restrict $|q_i|\leq \lfloor L/2 \rfloor$. The $(0,0)$-twisted state corresponds to complete phase synchronization. Figure \ref{fig:snapshot} shows an example of a $(3,2)$-twisted state solution.
	
	\begin{figure}[h]
		\centering
		\includegraphics[width = .45\textwidth]{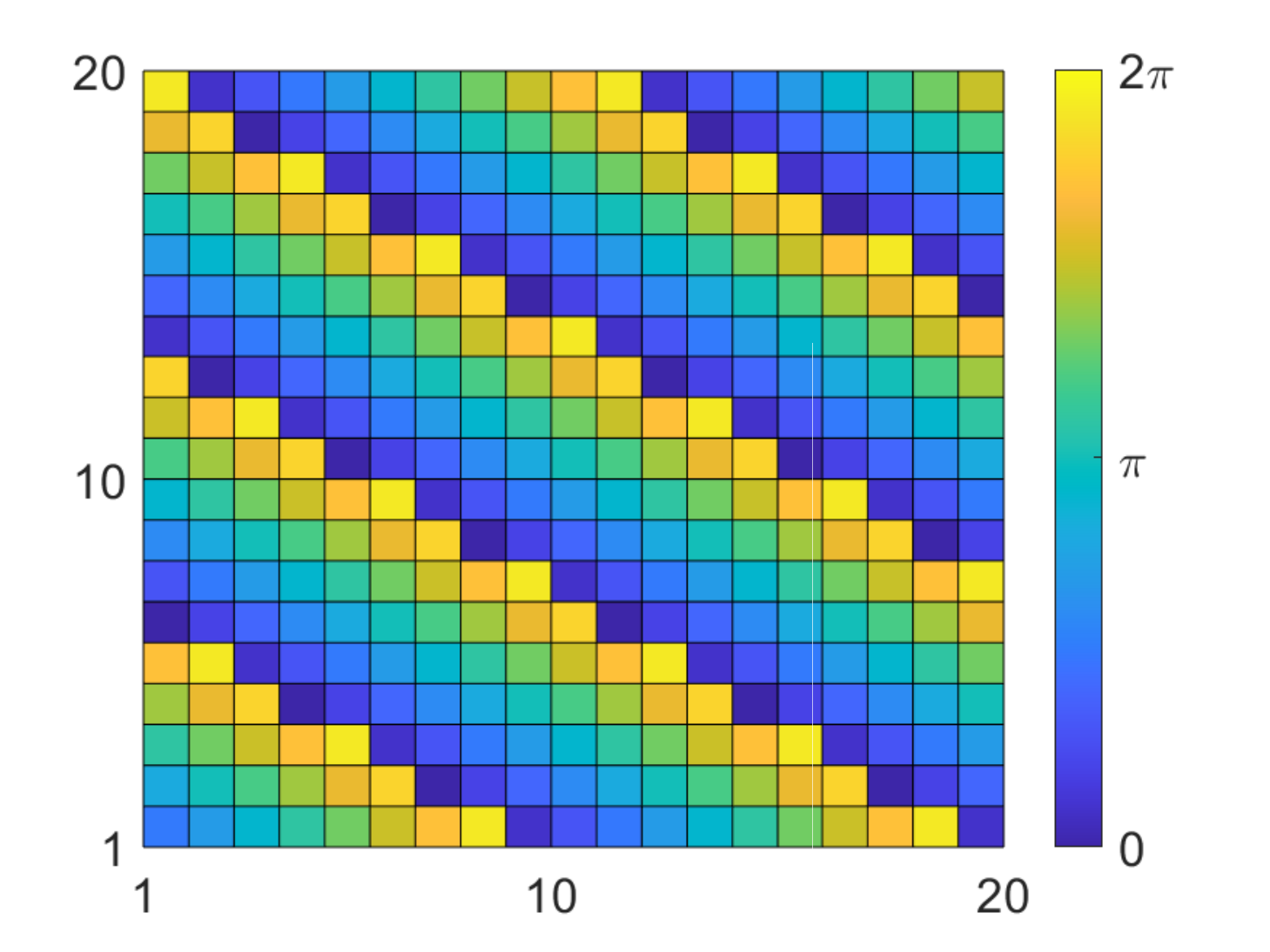}
		\caption{Snapshot of a $(3,2)$-twisted state equilibrium state on a lattice of size $20\times 20$.}
		\label{fig:snapshot}
	\end{figure}

	\begin{lemma}
		Given any coupling distance $r$, the $(q_1,q_2)$-twisted state is an equilibrium solution of \eqref{eq:kuramoto1}.
	\end{lemma}
	
	\begin{proof}
		Notice from the definition of $r$-nearest neighbor coupling that if $\theta_{i,j} \sim \theta_{i+\Delta_1,j+\Delta_2}$ then $\theta_{i,j} \sim \theta_{i-\Delta_1 ,j-\Delta_2 }$, where indices are interpreted modulo $L$ (e.g., if $i-\Delta_1 = 0$ then we interpret it as $i-\Delta_1 = L$). Then
		\begin{align*}
		\sin(\theta_{i+\Delta_1,j+\Delta_2} - \theta_{i,j}) &= \sin\left(\frac{2\pi q_1 \Delta_1}{L} + \frac{2\pi q_2 \Delta_2}{L}\right) \\
		& = -\sin\left(\frac{2\pi q_1 (-\Delta_1)}{L} + \frac{2\pi q_2 (-\Delta_2)}{L}\right)=-	\sin(\theta_{i-\Delta_1,j-\Delta_2} - \theta_{i,j}).
		\end{align*}
		Thus any non-zero terms in the right hand sum of \eqref{eq:kuramoto1} must come in equal and opposite pairs. 
	\end{proof}
	
	We opt to formulate the linear stability problem in terms of quadratic forms to both simplify and streamline the following proofs. This will also allow for more insight when extending to longer range interactions.
	
	\begin{lemma}\label{lem:jacobian}
		The Jacobian matrix of \eqref{eq:kuramoto} at an equilibrium $\bar{u} = (\bar{u}_1,\dots, \bar{u}_n)^T$ has associated quadratic form defined by
		\begin{equation}
		x^T J x = -\frac{1}{2}\sum_{j=1}^n\sum_{k\neq j} A_{j,k} \cos(\bar{u}_k-\bar{u}_j)(x_j-x_k)^2. \label{eq:quadraticform}
		\end{equation}
	\end{lemma}
	
	\begin{proof}
		Explicit calculation of the Jacobian of \eqref{eq:kuramoto} shows
		\begin{equation}
		J_{i,j} = \left\{ \begin{array}{cc}
		-\displaystyle\sum_{k=1}^n A_{i,k} \cos(\bar{u}_k - \bar{u}_i) &\mbox{if } i=j \\
		A_{i,j} \cos(\bar{u}_j-\bar{u}_i) & \mbox{if } i\neq j.
		\end{array}\right. 
		\end{equation}
		It follows that
		\begin{align*}
		x^T J x &= x^T \left(\begin{array}{c}
		-\displaystyle\sum_{k=1}^n A_{1,k} \cos(\bar{u}_k - \bar{u}_1)x_1 +\sum_{k\neq 1} A_{1,k} \cos(\bar{u}_k-\bar{u}_1)x_k\\
		\vdots \\
		-\displaystyle\sum_{k=1}^n A_{n,k} \cos(\bar{u}_k - \bar{u}_n)x_n +\sum_{k\neq n} A_{n,k} \cos(\bar{u}_k-\bar{u}_n)x_k
		\end{array}\right)\\
		&= \sum_{j=1}^n \sum_{k\neq j} A_{j,k} \cos(\bar{u}_k-\bar{u}_j)(-x_j^2+x_jx_k).
		\end{align*} 
		Since $A$ is symmetric and cosine is even, we can rearrange the sums to combine the $A_{j,k}$ and $A_{k,j}$ terms so
		\begin{align}
		x^T J x &= \frac{1}{2}\sum_{j=1}^n \sum_{k\neq j} A_{j,k} \cos(\bar{u}_k-\bar{u}_j)(-x_j^2+2x_jx_k-x_k^2) \\
		&= -\frac{1}{2}\sum_{j=1}^n \sum_{k\neq j} A_{j,k} \cos(\bar{u}_k-\bar{u}_j)(x_j-x_k)^2.
		\end{align}
	\end{proof}
	
	Since $J$ is symmetric its eigenvalues are real. We recall from standard results on quadratic forms that $x^T J x <0$ for all $x\in \mathbb{R}^n$ if and only if the eigenvalues of $J$ are all less than zero, and if $x^T J x$ is sometimes positive and sometimes negative, then $J$ has eigenvalues of both positive and negative sign. Moreover if $x$ is an eigenvector of $J$ with eigenvalue $\lambda$ then
	\begin{equation*}
	x^T J x = \lambda \|x\|^2.
	\end{equation*} 
	As such, we can deduce the stability of a given equilibrium by estimating the range of the quadratic form.
	
	
	\subsection{Nearest neighbor interactions}

	On 1D ring networks Wiley, Strogatz, and Girvan \cite{wiley2006size} established stability criteria for $q$-twisted states. Later Lee, Cho, and Hong \cite{lee2018twisted} established stability of $(q_1,0)$-twisted states for $r=1$ nearest neighbor graphs on 2D and 3D lattices. In the more general case of $(q_1,q_2)$-twisted states, stability analysis is similar in spirit to the latter work. However we believe our formulation using quadratic forms is more readily amenable to extensions toward longer range interactions and higher dimensions.
	
	\begin{theorem}\label{thm:nearest_neighbor}
		Let $1\leq r < \sqrt{2}$ (nearest neighbor coupling). The $(q_1,q_2)$-twisted state is asymptotically stable if
		\begin{equation*} 
		4\max\{|q_1|,|q_2|\}< L.
		\end{equation*}
		If $4\max\{|q_1|,|q_2|\}>L$ then the $(q_1,q_2)$-twisted state is unstable.
	\end{theorem}	
	
	\begin{proof}
		We begin by showing that $J$ is negative semi-definite. 
		For a twisted state with nearest neighbor coupling there are only two possibilities for adjacent terms:
		\begin{equation*}
		\bar{u}_k -\bar{u}_i = \left\{\begin{array}{cl}	
		\pm \frac{2\pi q_1}{L} & \mbox{for vertical neighbors} \\
		\pm\frac{2\pi q_2}{L} & \mbox{for horizontal neighbors}
		\end{array}\right.
		\end{equation*}
		The assumption $4\max\{|q_1|,|q_2|\}<L$ guarantees that $\cos(\bar{u}_k-\bar{u}_j)> 0$ and so every term in \eqref{eq:quadraticform} has definite sign and $x^TJx \leq 0$ for all $x$, and thus all eigenvalues of $J$ are non-positive.
		
		\medskip
		
		By taking $x = (1,\dots, 1)^T$ we see $Jx = 0$. This eigenvector corresponds to constant shifts of the equilibrium solution. To establish asymptotic stability, we show transverse stability to this $0$ eigenspace. To that end we show that this eigenvector spans the $0$ eigenspace by which we conclude all other eigenvalues must be strictly negative. Suppose $y\in\mathbb{R}^n$ is any $0$ eigenvector of $J$. Then $y^TJy=0$ so
		\begin{equation*}
		\sum_{j=1}^n\sum_{k\neq j} A_{j,k} \cos(\bar{u}_k-\bar{u}_j)(y_j-y_k)^2=0.
		\end{equation*}
		For any terms in the sum corresponding to adjacent oscillators, $A_{j,k}=1$ and so it follows that $y_j=y_k$, since $\cos(\bar{u}_k-\bar{u}_j)\neq 0$. Since the graph is connected $y_j \equiv C$ and thus the $0$ eigenspace is one dimensional. We conclude that all other eigenvalues are strictly negative and thus the equilibrium is stable.
		
		\medskip
		
		On the other hand suppose 	$4\max\{|q_1|,|q_2|\}> L$. Without loss of generality we assume $4q_1 > L$ (and recall that $|q_1| \leq \lfloor L/2 \rfloor$). Taking $x_1 = x_2= \cdots = x_L = 1$ and all other $x_i=0$ we obtain
		\begin{equation*}
		x^T J x = -\frac{1}{2}\sum_{j=1}^L 2\cos\left(\frac{2\pi q_1}{L}\right) = -L\cos\left(\frac{2\pi q_1}{L}\right) >0.
		\end{equation*}
		As such $J$ has positive eigenvalues and the twisted state is unstable. This construction exploits the fact that the phase-shifts in the vertical direction cause the instability and as such we choose terms $x_i$ corresponding to a single horizontal line of oscillators in the $L\times L$ array. If, on the other hand, $4q_2> L$ then one instead takes $x_i$ corresponding to a vertical line of oscillators in the array.
	\end{proof}
	
	\begin{remark}
		Linear stability analysis fails in the borderline case when $4\max\{|q_1|,|q_2|\} =L$ as the $0$ eigenspace becomes higher dimensional. As such one expects to require a more subtle center manifold analysis to establish the asymptotic dynamics. Numerical simulations suggest that these twisted states are unstable though we leave the rigorous analysis of this question for future work.
	\end{remark}
	
	\subsection{Diagonal interactions}
	
	In the range $\sqrt{2}\leq r < 2$ oscillators are connected to diagonal neighbors in addition to the nearest neighbors (see Figure \ref{fig:lattice}{\bf b}). Our previous analysis extends reasonably well, however our estimates are no longer sharp.
	
	\begin{theorem}\label{thm:diag_sufficient}
		Let $\sqrt{2}\leq r <2$. If $4(|q_1|+|q_2|)<L$ then the $(q_1,q_2)$-twisted state is asymptotically stable.
	\end{theorem}
	
	\begin{proof}
		As before, there are few cases for adjacent oscillators to check:
		\begin{equation*}
		\bar{u}_k -\bar{u}_i = \left\{\begin{array}{cl}	
		\pm \frac{2\pi q_1}{L} & \mbox{for vertical neighbors} \\
		\pm \frac{2\pi q_2}{L} & \mbox{for horizontal neighbors} \\
		\pm \frac{2\pi (q_1\pm q_2)}{L} & \mbox{for diagonal neighbors}.
		\end{array}\right.
		\end{equation*}
		A sufficient condition for stability is that $ \cos(\bar{u}_k-\bar{u}_i)>0$ for all $i,k$. That is, when $|\bar{u}_k-\bar{u}_i|<\frac{\pi}{2}$. This follows from the hypothesis since $|q_1\pm q_2|\leq |q_1|+|q_2|$. Thus $x^T J x$ is negative semi-definite. The graph of oscillators is connected so we can repeat the arguments of Theorem \ref{thm:nearest_neighbor} to conclude that the $0$-eigenspace is 1 dimensional and so the $(q_1,q_2)$-twisted state is asymptotically stable.	
	\end{proof}
	
	We now turn to instability. Recall that a necessary and sufficient condition for instability of a given twisted state is that $x^T J x> 0$ for some $x \in \mathbb{R}^n$. Since an exhaustive search is impossible, one instead hopes to exploit the coupling and twisted state geometry to instead choose $x$ in a clever way to maximize $x^T J x$, (cf. the choice of $x$ in the nearest neighbor case.)
	
	Difficulty here arises from the fact that, generally, each non-zero $x_i$ term in \eqref{eq:quadraticform} results in additional contributions to the sum from $\cos\left(\frac{2\pi q_1}{L}\right),$ $\cos\left(\frac{2\pi q_2}{L}\right)$, $\cos\left(\frac{2\pi (q_1+q_2)}{L}\right), $ and $\cos\left(\frac{2\pi (q_1-q_2)}{L}\right)$ values. These can have opposite signs resulting in non-obvious cancellation effects. Of course, a sufficient condition for instability is to require that each of these terms is negative, as this will unambiguously result in positive values of $x^T J x$. 
	
	However to improve such a crude bound, we notice that contributions in the sum from diagonal elements, will lead to instabilities before horizontal or vertical phase shifts alone (since the phase shifts from diagonal terms can be larger than nearest neighbors). So, by careful choice of $x_i$ we can explicitly cancel the contributions from the smaller of the diagonal phase shifts, resulting in a slightly stronger bound.
	
	\begin{theorem}\label{thm:diag_necessary}
		Let $\sqrt{2} \leq r < 2$. If $\cos\left(\frac{2\pi q_1}{L}\right) +\cos\left(\frac{2\pi q_2}{L}\right) +\cos\left(\frac{2\pi (|q_1|+|q_2|)}{L}\right)<0$ then the $(q_1,q_2)$-twisted state is unstable.
	\end{theorem}

	\begin{proof}
		If $q_1$ and $q_2$ have opposite signs then take 
		\begin{equation*}
		x_i = 1 \mbox{ if } i = k + (k-1)L,\;\; k\in\{1,\dots, L\}
		\end{equation*}
		and $x_i=0$ otherwise. That is $x_i=1$ if only if $i$ corresponds to an oscillator on the main diagonal of the $L\times L$ array of oscillators. On the other hand if $q_1$ and $q_2$ have the same sign then take 
		\begin{equation*}
		x_i = 1 \mbox{ if } i = kL - (k-1),\;\; k\in\{1,\dots, L\},
		\end{equation*} 
		and $x_i=0$ otherwise. Then $x_i=1$ if and only if $i$ corresponds to an oscillator on the off-diagonal. Plugging this into \eqref{eq:quadraticform}, we see that our choice of $x$ has resulted in the cancellation of the smaller of the two diagonal term contributions. Thus in either case we have
		\begin{equation*}
		x^T J x = -L\left( \cos\left(\frac{2\pi q_1}{L}\right) +\cos\left(\frac{2\pi q_2}{L}\right) +\cos\left(\frac{2\pi (|q_1|+|q_2|)}{L}\right)\right).
		\end{equation*}
		By assumption this term is positive and so $J$ has positive eigenvalues.
	\end{proof}
	
	In Figure \ref{fig:compare_theory} we plot the regions of guaranteed stability and instability on the $(q_1,q_2)$ plane, together with numerical simulations revealing the actual stability region.

	\begin{figure}[h]
		\centering
		{\bf a}\includegraphics[width = .45\textwidth]{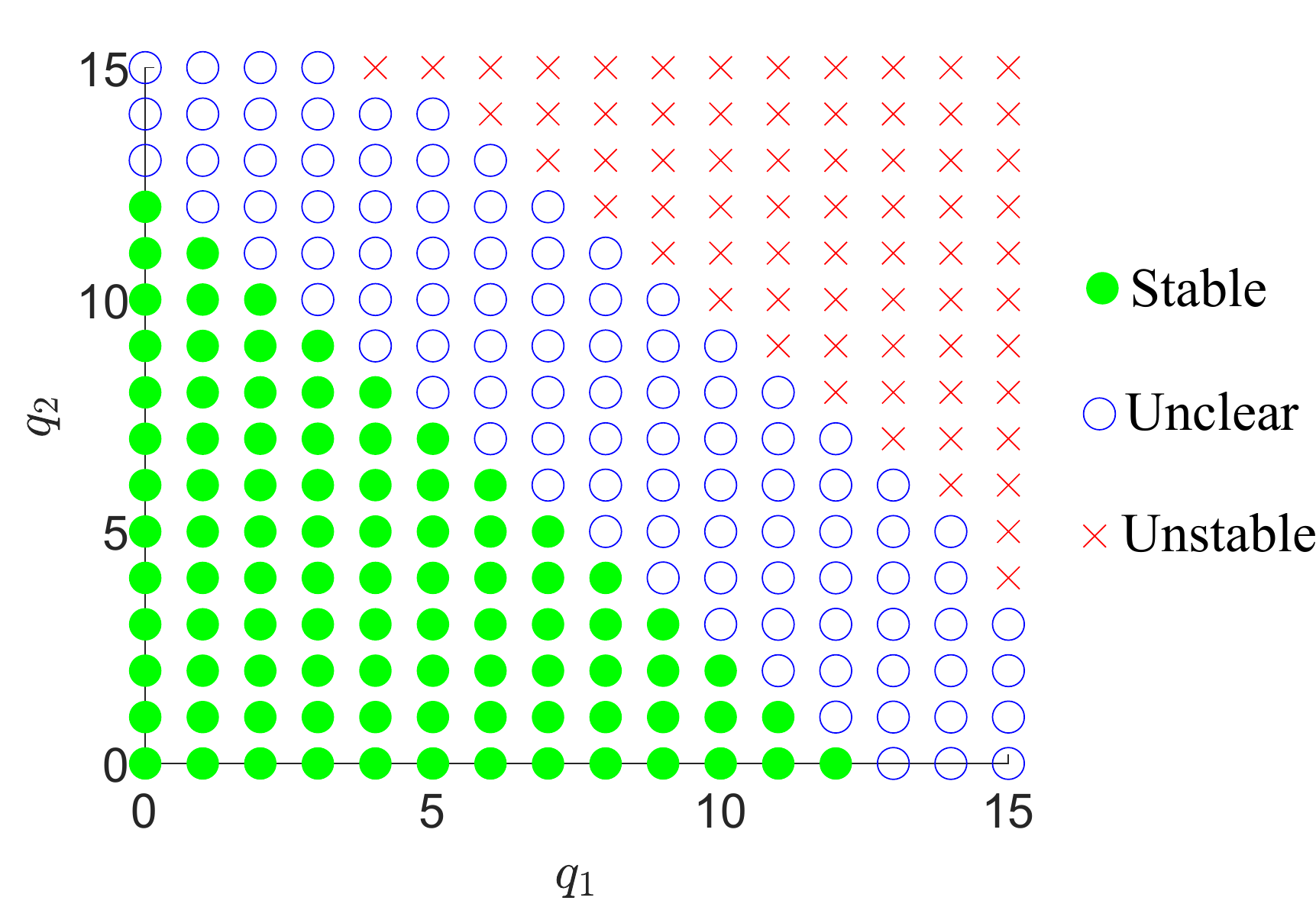} 
		{\bf b}\includegraphics[width = .45\textwidth]{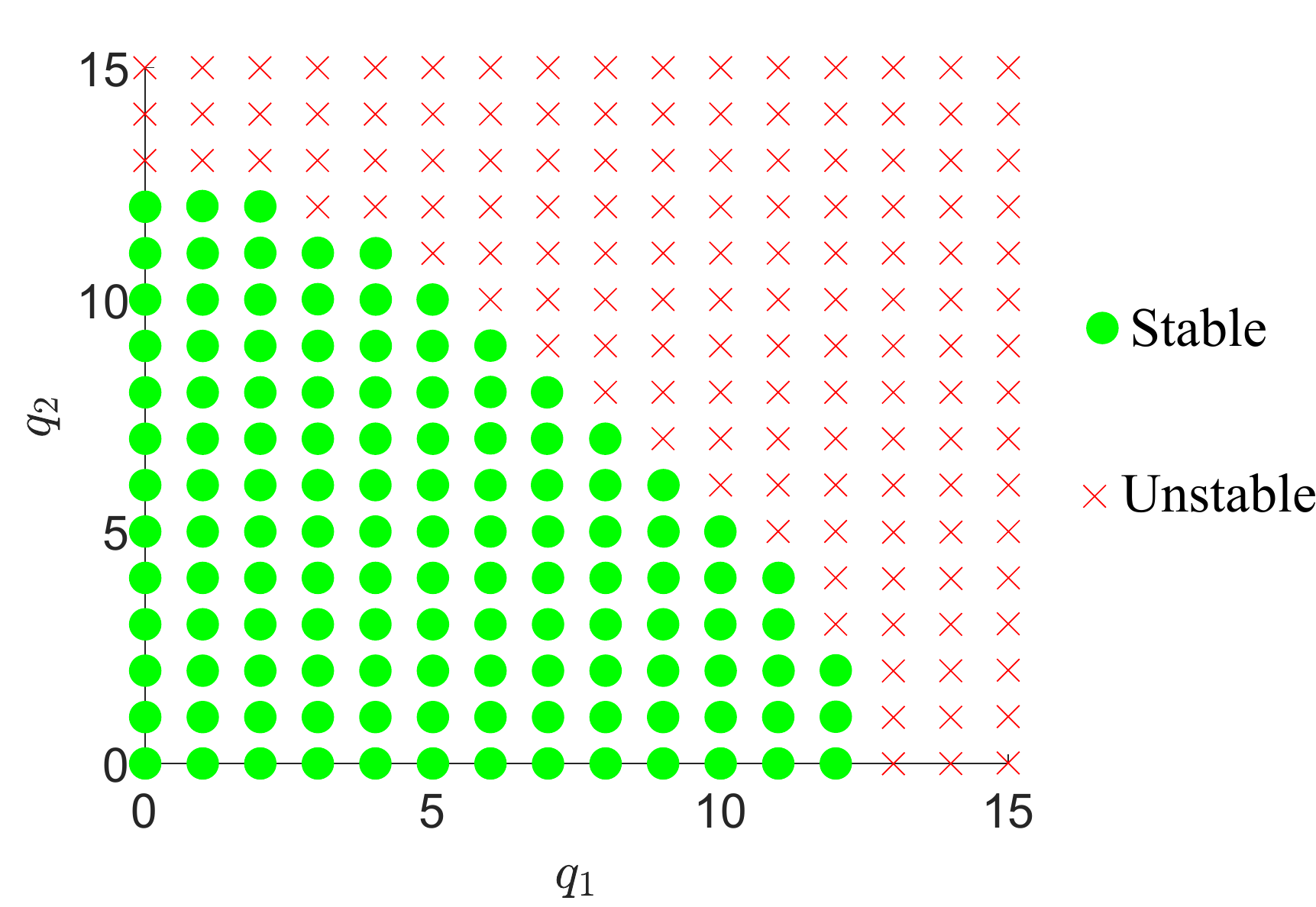}
		\caption{Stability of $(q_1,q_2)$-twisted states for a $52\times 52$ array of oscillators with $r=\sqrt{2}$ coupling. We compare ({\bf a}) theoretical bounds on the regions of stability and instability provided by Theorems \ref{thm:diag_sufficient} and \ref{thm:diag_necessary} and ({\bf b}) results from numerical simulations.}
		\label{fig:compare_theory}
	\end{figure}
	
	\subsection{Generalizing to longer ranges and higher dimensions} 
	
	As one increases $r$ and the dimension of the lattice, our previous arguments can be extended in a natural way. On an $m$-dimensional lattice we extend the definition of $r$-nearest neighbor coupling to be $mD$ Euclidean distance with periodic boundaries. Moreover the definition of $(q_1,q_2,\dots, q_m)$-twisted states extends in the natural way so that one generates independent twists in each dimension of the lattice. Next, notice that Lemma \ref{lem:jacobian} still applies since the only change mathematically is the structure of the adjacency matrix $A$, which is still symmetric. Thus we obtain the following result.
	
	\begin{theorem}\label{thm:general} Let $r\geq 1$. Then given $r$-nearest neighbor coupling on an $m$-dimensional lattice, the $(q_1,q_2,\dots q_m)$-twisted state is asymptotically stable if 
		\begin{equation}\label{eq:higherd}
		4\max_{(i_1,\dots, i_m) \in S_r}\left\{\sum_{j=1}^m i_j |q_i|\right\}<L
		\end{equation}
		where $S_r =\{(i_1,\dots, i_m) \colon \sum_{j=1}^m i_j^2 \leq r^2\}$.
	\end{theorem} 
	
	\begin{proof}
		Consider two oscillators on the lattice which are separated by $i_1$ steps in the first dimension, $i_2$ steps in the second dimension, and so on. These oscillators are adjacent if and only if $\sum_{j=1}^m i_j^2 \leq r^2$. Their relative phase difference is $\sum_{j=1}^m \frac{2\pi i_jq_i}{L}$. This corresponds to a potentially non-zero contribution in \eqref{eq:quadraticform} with coefficient
		\begin{equation}\label{eq:cosine_term}
		\cos\left( \sum_{j=1}^m \frac{2\pi i_jq_i}{L}\right),
		\end{equation}
		and since $\left|\sum_{j=1}^m i_jq_i\right| \leq \sum_{j=1}^m i_j|q_i| < L/4$, it follows that \eqref{eq:cosine_term} is positive. This argument holds for any adjacent oscillators and so $x^T J x\leq 0$ for all $x \in\mathbb{R}^n$. As before since the adjacency matrix $A$ corresponds to a connected graph, the $0$ eigenspace is spanned by $x=(1,1,\dots, 1)^T$ and thus all other eigenvalues of $J$ are strictly negative. 	
	\end{proof}

	We apply Theorem \ref{thm:general} to two examples. First, in the $m=2$ case with $r=2$ we obtain the stability criterion $4\max\{2|q_1|,2|q_2|,|q_1|+|q_2|\}<L$ (where some terms in $S_r$ have been trivially omitted). This can be further simplified to
	\begin{equation*} 
	8\max\{|q_1|,|q_2|\}<L.
	\end{equation*}

	In the $m=3$ case with $r=\sqrt{2}$ we obtain the stability criterion
	\begin{equation*} 
	4\max\{|q_1|+|q_2|,|q_2|+|q_3|,|q_1|+|q_3|\}<L.
	\end{equation*}
	Comparing these regions to the region of stability obtained from the Jacobian matrix eigenvalues shows that these estimates align with the true regions of stability reasonably well, see Figure \ref{fig:r2_bounds}. The stable twisted states outside of the bounds \eqref{eq:higherd} are more subtle due to additional cancellation effects in \eqref{eq:quadraticform} which our analysis does not account for. 
	
	\begin{figure}[h]
		\centering
		{\bf a}\includegraphics[width = .45\textwidth]{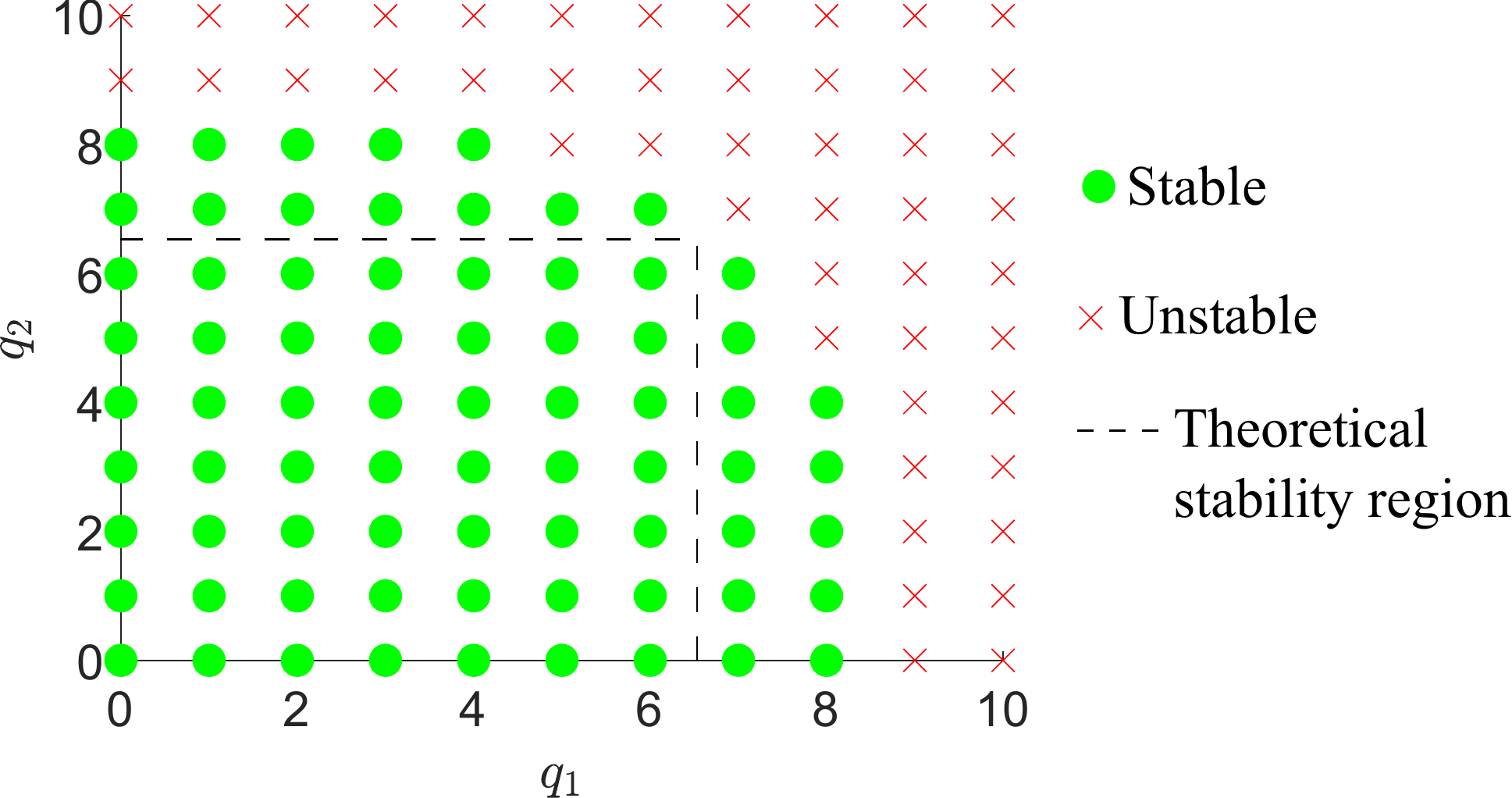}
		{\bf b}\includegraphics[width =.45\textwidth]{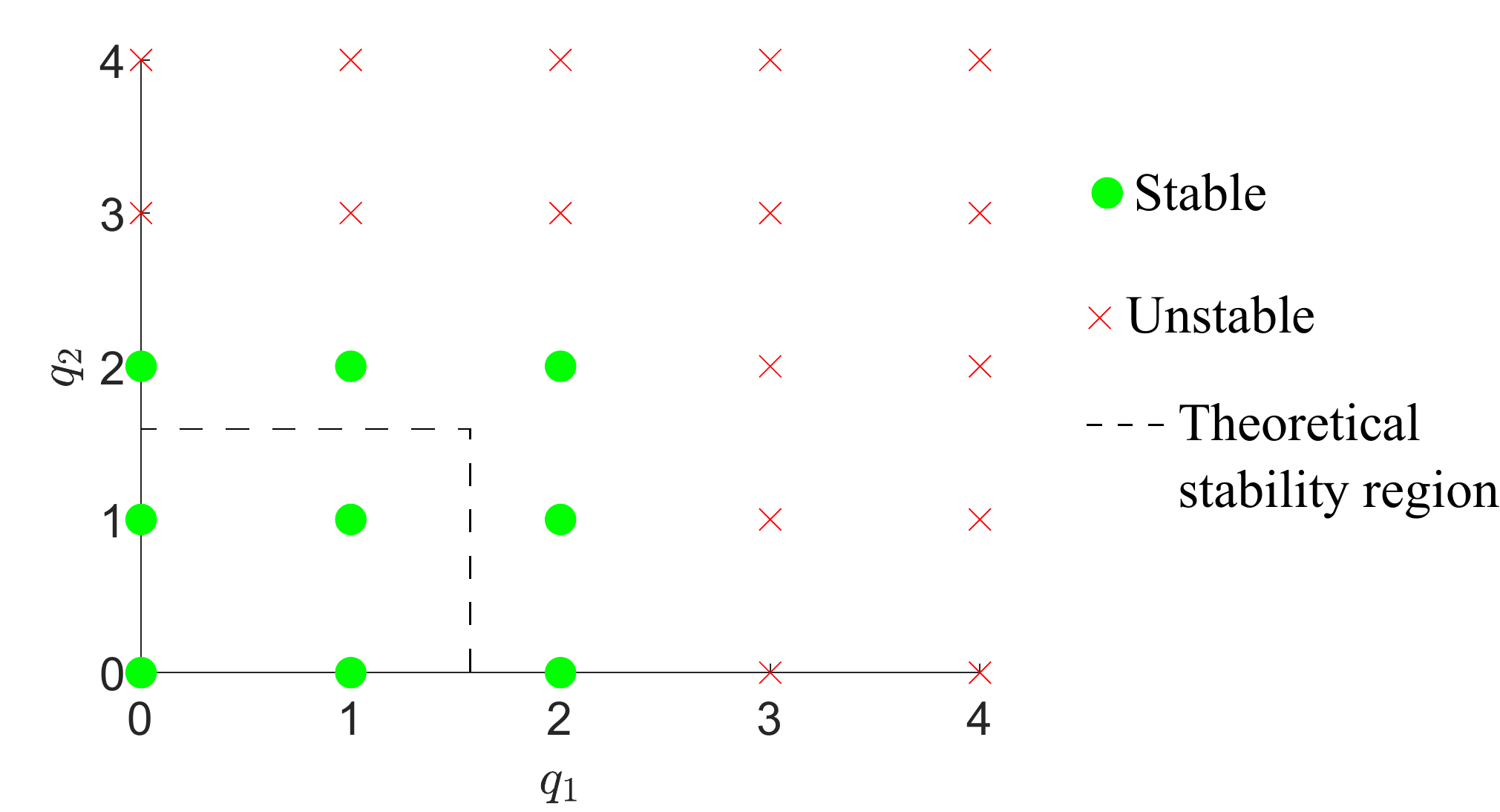}
		\caption{The region of stable twisted states guaranteed from Theorem \ref{thm:general} when compared to eigenvalues of Jacobian matrix. We test ({\bf a}) $r=2$ on a $52\times 52$ lattice and ({\bf b}) $r=\sqrt{2}$ on a $12\times 12\times 12$ lattice on the cross-section with fixed $q_3=1$.}
		\label{fig:r2_bounds}	
	\end{figure}

	Obtaining explicit estimates on the instability of twisted states is a more subtle business. In the case of nearest neighbor coupling one can take $x$ analogously to the construction in the second half of Theorem \ref{thm:nearest_neighbor}. Namely, one can let $x_j=1$ corresponding to a hyperplane of oscillators perpendicular to the twisted state with largest $q_i$. Then, all contributions to \eqref{eq:quadraticform} vanish except for exactly those terms containing $\cos\left(\frac{2\pi q_i}{L}\right)$, and so it follows for nearest neighbor coupling that the estimate \eqref{eq:higherd} is sharp in arbitrary dimensions.
	
	\section{Estimating eigenvalues of $J$}
	
	We end with an observation on the eigenvalues of $J$. In the study of twisted states on 1D ring coupling \cite{wiley2006size}, the Jacobian is a circulant matrix, and so its eigenvalues are explicit \cite{townsend2020dense}:
	\begin{align}
	\lambda_j &= -\sum_{j=1}^n A_{1,k} \cos(\bar{u}_k-\bar{u}_1) + \omega^j A_{1,2} \cos(\bar{u}_2-\bar{u}_1) + \cdots + \omega^{(n-1)j}A_{1,n} \cos(\bar{u}_n-\bar{u}_1)\\
	&= \sum_{k=1}^n A_{1,k}\cos(\bar{u}_k-\bar{u}_1)(\omega^{(k-1)j}-1),
	\end{align}
	where $\omega = \exp(2\pi i/n)$ is a primitive $n$th root of unity. Coupling on higher dimensional lattices results in a matrix which is no longer circulant. For example on a 2D lattice of size $L\times L$, many oscillators $u_j$ have right nearest neighbor $u_{j+1}$, but at $u_L$ the right nearest neighbor is instead $u_1$. This breaks the circulant nature of the matrix. 
	In general, in an $m$-dimensional lattice with side length $L$, such edge cases correspond to $O(L^{m-1})$ entries in the $L^m\times L^m$ Jacobian matrix. Thus it is not unreasonable to treat $J$ as ``almost'' circulant. 

	More rigorously, let $C$ be the circulant matrix which is obtained from $J$ by changing the fewest number of entries. Denote the eigenvalues of $C$ and $J$ as $\lambda^C_i$ and $\lambda^J_i$ respectively and assume they are arranged in increasing order. Then the Hoffman-Wielandt inequality provides an estimate on the average error between the eigenvalues:
	\begin{equation*}
	\frac{1}{n}\sum_{i=1}^n |\lambda^J_i - \lambda^C_i|^2 \leq \frac{1}{n}\|J-C\|^2_F,
	\end{equation*}
	where $\|A\|^2_{F} = \sum_{i,j} |A_{i,j}|^2$ is the Frobenius norm. Since there are $O(L^{m-1})$ entries which must be changed to make $J$ circulant, then 
	\begin{equation*}
	\frac{1}{n}\sum_{j=1}^n |\lambda^J_i - \lambda^C_i|^2 = O(L^{m-1}/n)= O(1/L).
	\end{equation*}
	
	Moreover, this motivates the hypothesis that the eigenvectors of $J$ may be well approximated by eigenvectors for circulant matrices:
	\begin{equation} 
	\nu_i := \left(1,\cos\left(\frac{2\pi}{n}i\right), \dots, \cos\left(\frac{2\pi}{n}(n-1)i\right)\right),\;\;\; i \in \{1,\dots, n\}.
	\end{equation}
	Writing
	\begin{equation*}
	\bar{\lambda}_i^J = \frac{\nu_i^T J \nu_i}{\|\nu_i\|^2},
	\end{equation*}
	Figure \ref{fig:circulant} shows that, indeed, the average error between $\nu_i^T J \nu_i$ and the eigenvalues of $J$ is quite small and improves as $L$ increases.
	
	In general, numerical tests show that $\nu_i^T J \nu_i$ can in most situations accurately predict stability of the twisted state (errors arise, e.g., on boundary cases when the 0 eigenspace is higher dimensional, since the $x_\ell$ is never truly a $0$ eigenvector). While we were unable to translate this idea to explicit analytical estimates on stability of twisted states, we believe it may be a fruitful endeavor for future work to improve bounds on regions of instability, especially for large $L$. That is, finding an $x_i$ to maximize the quadratic form \eqref{eq:quadraticform} is in general quite challenging, but evidently by choosing $\nu_i$ optimally one can aim to maximize the value $\nu_i^T J \nu_i$ and find sufficient conditions for instability of twisted states. 
	
	\begin{figure}
		{\bf a} \includegraphics[width = .45\textwidth]{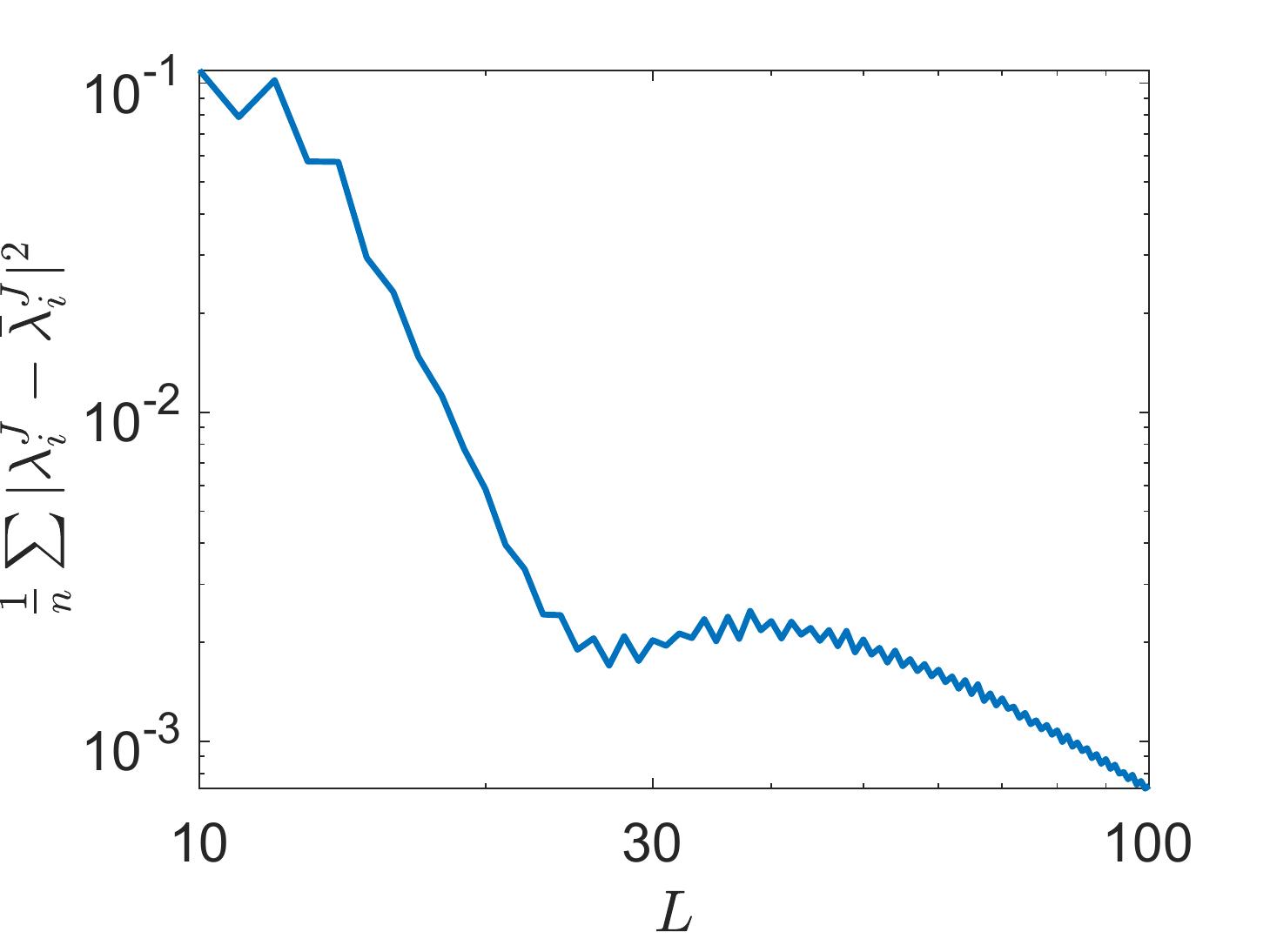}
		{\bf b} \includegraphics[width = .45\textwidth]{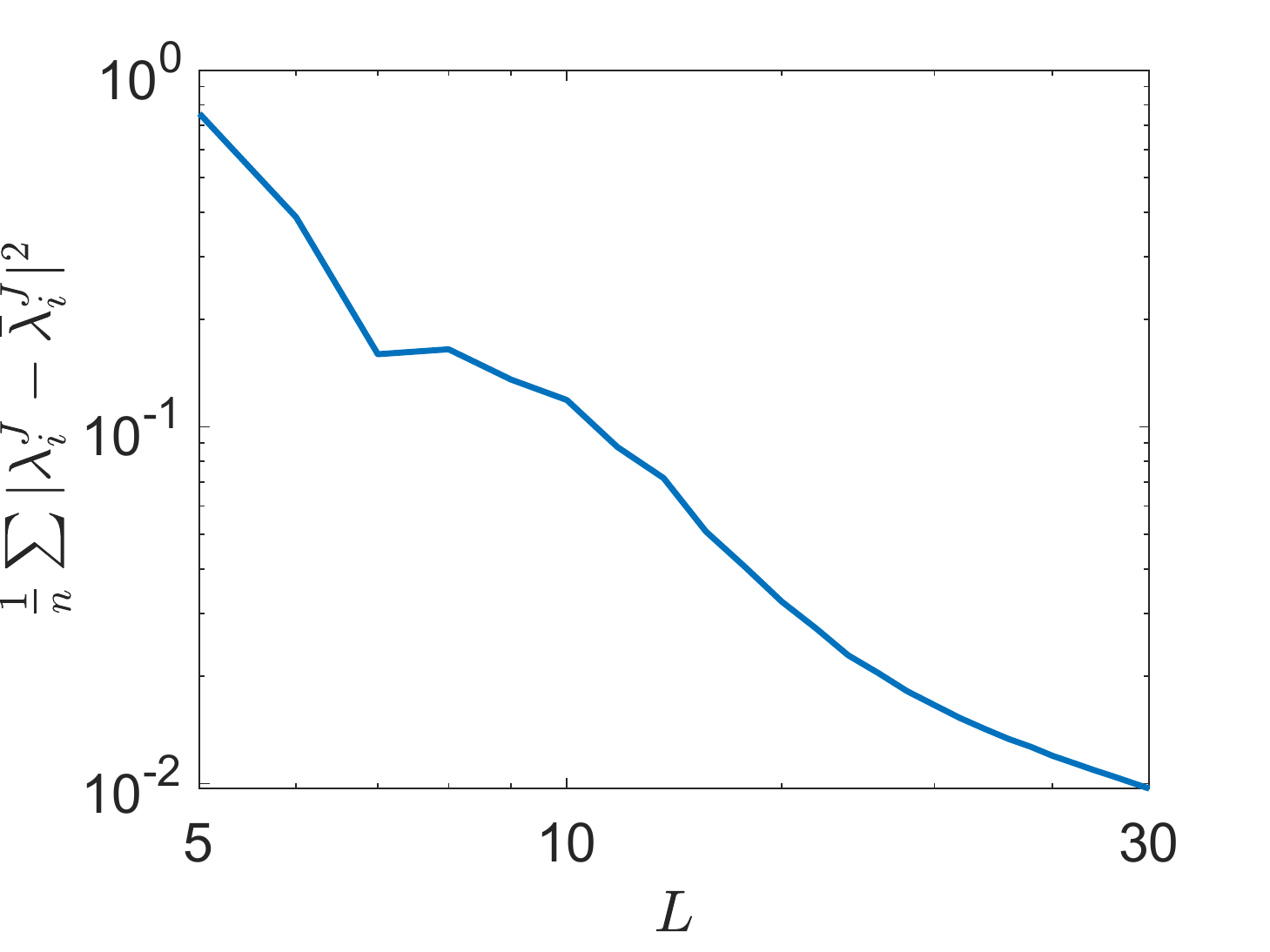}
		\caption{Average error between the eigenvalues of the Jacobian, $\lambda^J_i$, and their approximations $\bar{\lambda}^J_i$. We observe error decreases for both ({\bf a}) an $L\times L$ lattice with $r=2$ and a $(3,2)$-twisted state and ({\bf b}) an $L\times L\times L$ lattice with $r=\sqrt{2}$ and $(3,2,1)$-twisted state.}
		\label{fig:circulant}
	\end{figure}
	%
	%

	\section{Conclusion}

In this work we have studied linear stability of twisted states on lattices of Kuramoto oscillators. Previous work established stability of twisted states in 1D \cite{wiley2006size} and on nearest neighbor coupling in 2D and 3D \cite{lee2018twisted}, but our results are the first to our knowledge studying non-local interactions in arbitrary dimensions. Establishing linear stability of twisted states relied on finding either lower or upper bounds on the quadratic form associated to the Jacobian. These estimates are sharp for nearest neighbor lattices in arbitrary dimensions. To tighten estimates on lattices with longer range interactions, one must better exploit the non-trivial cancellation effects in \eqref{eq:quadraticform}. We also reported an interesting connection to eigenvectors of circulant matrices. Converting this numerical observation to analytical estimates on the regions of stability and instability remains elusive and a goal of future studies. Additionally, we propose potentially interesting work to be done in the borderline cases where the linear stability analysis fails, as well as in understanding the basin of attraction of the twisted state solutions\cite{wiley2006size,lee2018twisted}.
	
	\medskip
	
	\noindent{\bf Acknowledgements.} Numerical simulations were completed using the high performance computing cluster (ELSA) at The College of New Jersey. Funding of ELSA is provided in part by National Science Foundation OAC-1828163. MSM was additionally endorsed by a Support of Scholarly Activities Grant at The College of New Jersey. The authors would like to thank Dr. Nicholas Battista for fruitful conversations and for his guidance in parallel processing on ELSA.



\begin{thebibliography}{10}
		
		\bibitem{acebonluivicperfelspi05}
		Juan~A Acebr{\'o}n, Luis~L Bonilla, Conrad J~P{\'e}rez Vicente, F{\'e}lix
		Ritort, and Renato Spigler.
		\newblock The {K}uramoto model: A simple paradigm for synchronization
		phenomena.
		\newblock {\em Reviews of modern physics}, 77(1):137, 2005.
		
		\bibitem{aoyagi1991frequency}
		Toshio Aoyagi and Yoshiki Kuramoto.
		\newblock Frequency order and wave patterns of mutual entrainment in
		two-dimensional oscillator lattices.
		\newblock {\em Physics Letters A}, 155(6-7):410--414, 1991.
		
		\bibitem{belykh2003persistent}
		I~Belykh, V~Belykh, K~Nevidin, and M~Hasler.
		\newblock Persistent clusters in lattices of coupled nonidentical chaotic
		systems.
		\newblock {\em Chaos: An Interdisciplinary Journal of Nonlinear Science},
		13(1):165--178, 2003.
		
		\bibitem{belykh2003cluster}
		Vladimir~N Belykh, Igor~V Belykh, Martin Hasler, and Konstantin~V Nevidin.
		\newblock Cluster synchronization in three-dimensional lattices of diffusively
		coupled oscillators.
		\newblock {\em International Journal of Bifurcation and Chaos},
		13(04):755--779, 2003.
		
		\bibitem{bolotov2019twisted}
		Dmitry Bolotov, Maxim Bolotov, Lev Smirnov, Grigory Osipov, and Arkady
		Pikovsky.
		\newblock Twisted states in a system of nonlinearly coupled phase oscillators.
		\newblock {\em Regular and Chaotic Dynamics}, 24(6):717--724, 2019.
		
		\bibitem{bukh2019spiral}
		Andrei Bukh, Galina Strelkova, and Vadim Anishchenko.
		\newblock Spiral wave patterns in a two-dimensional lattice of nonlocally
		coupled maps modeling neural activity.
		\newblock {\em Chaos, Solitons \& Fractals}, 120:75--82, 2019.
		
		\bibitem{chiba2018bifurcations}
		Hayato Chiba, Georgi~S Medvedev, and Matthew~S Mizuhara.
		\newblock Bifurcations in the kuramoto model on graphs.
		\newblock {\em Chaos: An Interdisciplinary Journal of Nonlinear Science},
		28(7):073109, 2018.
		
		\bibitem{corral1995self}
		{\'A}lvaro Corral, Conrad~J P{\'e}rez, Albert Diaz-Guilera, and Alex Arenas.
		\newblock Self-organized criticality and synchronization in a lattice model of
		integrate-and-fire oscillators.
		\newblock {\em Physical review letters}, 74(1):118, 1995.
		
		\bibitem{daido1988lower}
		Hiroaki Daido.
		\newblock Lower critical dimension for populations of oscillators with randomly
		distributed frequencies: a renormalization-group analysis.
		\newblock {\em Physical review letters}, 61(2):231, 1988.
		
		\bibitem{deville2019synchronization}
		Lee DeVille.
		\newblock Synchronization and stability for quantum kuramoto.
		\newblock {\em Journal of Statistical Physics}, 174(1):160--187, 2019.
		
		\bibitem{girnyk2012multistability}
		Taras Girnyk, Martin Hasler, and Yuriy Maistrenko.
		\newblock Multistability of twisted states in non-locally coupled kuramoto-type
		models.
		\newblock {\em Chaos: An Interdisciplinary Journal of Nonlinear Science},
		22(1):013114, 2012.
		
		\bibitem{hagerstrom2012experimental}
		Aaron~M Hagerstrom, Thomas~E Murphy, Rajarshi Roy, Philipp H{\"o}vel, Iryna
		Omelchenko, and Eckehard Sch{\"o}ll.
		\newblock Experimental observation of chimeras in coupled-map lattices.
		\newblock {\em Nature Physics}, 8(9):658, 2012.
		
		\bibitem{hong2005collective}
		H~Hong, Hyunggyu Park, and MY~Choi.
		\newblock Collective synchronization in spatially extended systems of coupled
		oscillators with random frequencies.
		\newblock {\em Physical Review E}, 72(3):036217, 2005.
		
		\bibitem{kuramoto1975self}
		Yoshiki Kuramoto.
		\newblock Self-entrainment of a population of coupled non-linear oscillators.
		\newblock In {\em International symposium on mathematical problems in
			theoretical physics}, pages 420--422. Springer, 1975.
		
		\bibitem{laing2009dynamics}
		Carlo~R Laing.
		\newblock The dynamics of chimera states in heterogeneous kuramoto networks.
		\newblock {\em Physica D: Nonlinear Phenomena}, 238(16):1569--1588, 2009.
		
		\bibitem{laing2011fronts}
		Carlo~R Laing.
		\newblock Fronts and bumps in spatially extended kuramoto networks.
		\newblock {\em Physica D: Nonlinear Phenomena}, 240(24):1960--1971, 2011.
		
		\bibitem{lee2018twisted}
		Seungjae Lee, Young~Sul Cho, and Hyunsuk Hong.
		\newblock Twisted states in low-dimensional hypercubic lattices.
		\newblock {\em Physical Review E}, 98(6):062221, 2018.
		
		\bibitem{lee2010vortices}
		Tony~E Lee, Heywood Tam, G~Refael, Jeffrey~L Rogers, and MC~Cross.
		\newblock Vortices and the entrainment transition in the two-dimensional
		kuramoto model.
		\newblock {\em Physical Review E}, 82(3):036202, 2010.
		
		\bibitem{medvedev2014small}
		Georgi~S Medvedev.
		\newblock Small-world networks of kuramoto oscillators.
		\newblock {\em Physica D: Nonlinear Phenomena}, 266:13--22, 2014.
		
		\bibitem{medvedev2021chimeras}
		Georgi~S Medvedev and Matthew~S Mizuhara.
		\newblock Chimeras unfolded.
		\newblock {\em arXiv preprint arXiv:2105.07541}, 2021.
		
		\bibitem{medtan15}
		Georgi~S Medvedev and Xuezhi Tang.
		\newblock Stability of twisted states in the {K}uramoto model on {C}ayley and
		random graphs.
		\newblock {\em Journal of Nonlinear Science}, 25(6):1169--1208, 2015.
		
		\bibitem{migliorini2008critical}
		Gabriele Migliorini.
		\newblock The critical properties of two-dimensional oscillator arrays.
		\newblock {\em Journal of Physics A: Mathematical and Theoretical},
		41(32):324021, 2008.
		
		\bibitem{nkomo2013chimera}
		Simbarashe Nkomo, Mark~R Tinsley, and Kenneth Showalter.
		\newblock Chimera states in populations of nonlocally coupled chemical
		oscillators.
		\newblock {\em Physical review letters}, 110(24):244102, 2013.
		
		\bibitem{omel2014partially}
		Oleh~E Omel'chenko, Matthias Wolfrum, and Carlo~R Laing.
		\newblock Partially coherent twisted states in arrays of coupled phase
		oscillators.
		\newblock {\em Chaos: An Interdisciplinary Journal of Nonlinear Science},
		24(2):023102, 2014.
		
		\bibitem{omel2012stationary}
		Oleh~E Omel'chenko, Matthias Wolfrum, Serhiy Yanchuk, Yuri~L Maistrenko, and
		Oleksandr Sudakov.
		\newblock Stationary patterns of coherence and incoherence in two-dimensional
		arrays of non-locally-coupled phase oscillators.
		\newblock {\em Physical Review E}, 85(3):036210, 2012.
		
		\bibitem{ostborn2009renormalization}
		Per {\"O}stborn.
		\newblock Renormalization of oscillator lattices with disorder.
		\newblock {\em Physical Review E}, 79(5):051114, 2009.
		
		\bibitem{ottino2016frequency}
		Bertrand Ottino-L{\"o}ffler and Steven~H Strogatz.
		\newblock Frequency spirals.
		\newblock {\em Chaos: An Interdisciplinary Journal of Nonlinear Science},
		26(9):094804, 2016.
		
		\bibitem{paullet1994stable}
		Joseph~E Paullet and G~Bard Ermentrout.
		\newblock Stable rotating waves in two-dimensional discrete active media.
		\newblock {\em SIAM Journal on Applied Mathematics}, 54(6):1720--1744, 1994.
		
		\bibitem{pikovsky2003synchronization}
		Arkady Pikovsky, Michael Rosenblum, Jurgen Kurths, and J{\"u}rgen Kurths.
		\newblock {\em Synchronization: a universal concept in nonlinear sciences},
		volume~12.
		\newblock Cambridge university press, 2003.
		
		\bibitem{rodrigues2016kuramoto}
		Francisco~A Rodrigues, Thomas K~DM Peron, Peng Ji, and J{\"u}rgen Kurths.
		\newblock The kuramoto model in complex networks.
		\newblock {\em Physics Reports}, 610:1--98, 2016.
		
		\bibitem{sakaguchi1987local}
		Hidetsugu Sakaguchi, Shigeru Shinomoto, and Yoshiki Kuramoto.
		\newblock Local and grobal self-entrainments in oscillator lattices.
		\newblock {\em Progress of Theoretical Physics}, 77(5):1005--1010, 1987.
		
		\bibitem{salova2020decoupled}
		Anastasiya Salova and Raissa~M D'Souza.
		\newblock Decoupled synchronized states in networks of linearly coupled limit
		cycle oscillators.
		\newblock {\em Physical Review Research}, 2(4):043261, 2020.
		
		\bibitem{sarkar2021phase}
		Mrinal Sarkar and Neelima Gupte.
		\newblock Phase synchronization in the two-dimensional kuramoto model: Vortices
		and duality.
		\newblock {\em Physical Review E}, 103(3):032204, 2021.
		
		\bibitem{shepelev2020quantifying}
		Igor~A Shepelev, Andrei~V Bukh, Sishu~S Muni, and Vadim~S Anishchenko.
		\newblock Quantifying the transition from spiral waves to spiral wave chimeras
		in a lattice of self-sustained oscillators.
		\newblock {\em Regular and Chaotic Dynamics}, 25(6):597--615, 2020.
		
		\bibitem{shepelev2018double}
		Igor~A Shepelev, AV~Bukh, Tatiana~E Vadivasova, Vadim~S Anishchenko, and Anna
		Zakharova.
		\newblock Double-well chimeras in 2d lattice of chaotic bistable elements.
		\newblock {\em Communications in Nonlinear Science and Numerical Simulation},
		54:50--61, 2018.
		
		\bibitem{shima2004rotating}
		Shin-ichiro Shima and Yoshiki Kuramoto.
		\newblock Rotating spiral waves with phase-randomized core in nonlocally
		coupled oscillators.
		\newblock {\em Physical Review E}, 69(3):036213, 2004.
		
		\bibitem{strogatz2000kuramoto}
		Steven~H Strogatz.
		\newblock From kuramoto to crawford: exploring the onset of synchronization in
		populations of coupled oscillators.
		\newblock {\em Physica D: Nonlinear Phenomena}, 143(1-4):1--20, 2000.
		
		\bibitem{strogatz2012sync}
		Steven~H Strogatz.
		\newblock {\em Sync: How order emerges from chaos in the universe, nature, and
			daily life}.
		\newblock Hachette UK, 2012.
		
		\bibitem{strogatz1988phase}
		Steven~H Strogatz and Renato~E Mirollo.
		\newblock Phase-locking and critical phenomena in lattices of coupled nonlinear
		oscillators with random intrinsic frequencies.
		\newblock {\em Physica D: Nonlinear Phenomena}, 31(2):143--168, 1988.
		
		\bibitem{totz2018spiral}
		Jan~Frederik Totz, Julian Rode, Mark~R Tinsley, Kenneth Showalter, and Harald
		Engel.
		\newblock Spiral wave chimera states in large populations of coupled chemical
		oscillators.
		\newblock {\em Nature Physics}, 14(3):282--285, 2018.
		
		\bibitem{townsend2020dense}
		Alex Townsend, Michael Stillman, and Steven~H Strogatz.
		\newblock Dense networks that do not synchronize and sparse ones that do.
		\newblock {\em Chaos: An Interdisciplinary Journal of Nonlinear Science},
		30(8):083142, 2020.
		
		\bibitem{wiley2006size}
		Daniel~A Wiley, Steven~H Strogatz, and Michelle Girvan.
		\newblock The size of the sync basin.
		\newblock {\em Chaos: An Interdisciplinary Journal of Nonlinear Science},
		16(1):015103, 2006.
		
	\end{thebibliography}
\end{document}